\documentclass[letterpaper, 10 pt, conference]{ieeeconf}  % Comment this line out
\IEEEoverridecommandlockouts                              % This command is only
\title{\LARGE \bf
A Study of an Atomic Mobility Game With Uncertainty Under Prospect Theory
}

\author{Ioannis Vasileios Chremos$^1$, Heeseung Bang$^{1,2}$, Aditya Dave$^2$, Viet-Anh Le$^{1,3}$, \textit{Student Members, IEEE} \\ Andreas A. Malikopoulos$^2$, \textit{Senior Member, IEEE}% <-this % stops a space
\thanks{This work was supported in part by NSF under Grants CNS-2149520 and CMMI-2219761, the Delaware Department of Transportation, and by the University of Delaware Graduate College through the University Doctoral Fellowship Award.}% <-this % stops a space
\thanks{$^{1}$Department of Mechanical Engineering, University of Delaware, Newark, DE 19716 USA. {\tt\footnotesize{ichremos@udel.edu}}}
\thanks{$^{2}$School of Civil and Environmental Engineering, Cornell University, Ithaca, NY 14850 USA.
{\tt\footnotesize{\{hb489,a.dave,amaliko\}@cornell.edu}}}
\thanks{$^{3}$Systems Engineering, Cornell University, Ithaca, NY 14850 USA. {\tt\footnotesize{vl299@cornell.edu}}}
}%

\usepackage{cite}
\usepackage{graphicx}
\usepackage{xcolor}
\usepackage{amsmath,amssymb}
\usepackage{dsfont}
\usepackage{comment}
\usepackage{caption}
\usepackage{subcaption}

\usepackage{algorithm}
\usepackage{algpseudocode}

\usepackage{amsfonts}

\usepackage{amsthm}

\usepackage[color=green!40]{todonotes}

\theoremstyle{plain}

\newtheorem{lemma}{Lemma}

\newtheorem{theorem}{Theorem}

\theoremstyle{definition}
\newtheorem{definition}{Definition}

\newtheorem{remark}{Remark}

\begin{document}

\maketitle

\thispagestyle{empty}
\pagestyle{empty}

\begin{abstract}
    In this paper, we present a study of a mobility game with uncertainty in the decision-making of travelers and incorporate prospect theory to model travel behavior. We formulate a mobility game that models how travelers distribute their traffic flows in a transportation network with splittable traffic, utilizing the Bureau of Public Roads function to establish the relationship between traffic flow and travel time cost. Given the inherent non-linearities and complexity introduced by the uncertainties, we propose a smooth approximation function to estimate the prospect-theoretic cost functions. As part of our analysis, we characterize the best-fit parameters and derive an upper bound for the error. We then show the existence of an equilibrium and its its best-possible approximation.
\end{abstract}

\section{INTRODUCTION}

%\subsection{Motivation}

Emerging mobility systems (e.g., connected and automated vehicles (CAVs), shared mobility) provide the most intriguing opportunity for enabling users to monitor transportation network conditions better and make efficient decisions for improving safety and transportation efficiency. The data and shared information of emerging mobility systems are associated with a new level of complexity in modeling and control \cite{Malikopoulos2016b}. The impact of selfish or irrational social behavior in routing networks of cars has been studied in recent years \cite{Mehr2019,Lazar2021,Biyik2021}. Other efforts have addressed \emph{how people learn and make routing decisions} with behavioral dynamics \cite{Krichene2015}. The problem of how travelers often have to make decisions under the uncertainty of experiencing delays, especially when uncertainties directly affect travel time in a transportation network, has not been adequately approached yet. Hence, our problem of interest is to study in a game-theoretic setting these interactions and analyze the equilibrium of the travelers' decisions under uncertainties \cite{chremos2022CSMArticle}. We study the interactions of a finite group of players that seek to travel in a transportation network (with a unique origin-destination pair) comprised of roads with splittable traffic. A key characteristic of our approach is that we incorporate \emph{prospect theory}, a behavioral model that captures the perceptions of utility under uncertainty (how likely and how much).

Some of the existing game-theoretical literature in control and transportation theory assumes that the players' behavior follows the \emph{rational choice theory}, i.e., each player is a risk-neutral and utility maximizer. This makes transportation models unrealistic, as unexpected travel delays can lead to uncertainty in a traveler's utility. There is strong evidence from empirical experiments that show how humans' choices and preferences systematically may deviate from the choice and preferences of a game-theoretic player under the rational choice theory \cite{Kahneman1979}. For example, humans compare the outcomes of their choices to a known expected amount of utility (called reference) and make their final decision, using that reference to assess their losses or gains asymmetrically. \emph{Prospect theory} has laid down the theoretical foundations to study such biases and the subjective perception of risk in the utility of humans \cite{Kahneman1979,Tversky1992}. This theory has been recognized as a closer-to-reality behavioral model for the decision-making of humans in different engineering problems \cite{Nar2017,chremos2022Prospect,Etesami2020}.

In general, one of the standard approaches to alleviate congestion in a transportation system has been managing the travel demand and supply while also taking into consideration the scarce resources. Such approaches focus primarily on traffic routing, which aims to optimize the routing decisions in a transportation network \cite{Silva2013}. Another approach is game theory that allows us to investigate the impact of selfish routing on efficiency and congestion \cite{Marden2015} and assign travelers routes to minimize travel time under a Nash Equilibrium (NE) \cite{Brown2017,chremos2020MobilityMarket,chremos2021MobilityGame,Chremos2020SocialDilemma,chremos2020SharedMobility}. A fundamental theoretical approach in alleviating congestion is \emph{routing/congestion games} \cite{Rosenthal1973,Rosen1965}, which are a generalization of the standard resource-sharing game of an arbitrary number of resources in a network.

%\subsection{Contribution of the Paper}

In this paper, we use Prelec's probability weighting function and an S-shaped value function to model how travelers perceive traffic uncertainties and their travel gains/losses. So, our first contribution is incorporating prospect theory into an atomic routing game with splittable traffic to capture a realistic version of the travelers' decision-making regarding travel time costs. The S-shaped value function is adopted to represent the curvature of the travel time cost function and account for the travelers' perception of gains/losses in travel time according to a reference point (defined using the US Bureau of Public Roads function). To address prospect theory's mathematical intractabilities, our second contribution proposes a smooth approximation function that estimates the non-linear piecewise prospect-theoretic cost functions. Thus, we can estimate how travelers perceive gains/losses and probabilities in travel time costs. This work is focused on establishing the fitness of the approximation function, proving the existence of at least one NE in pure strategies.

%\subsection{Organization of the Paper}

The remainder of the paper is structured as follows. In Section \ref{SECTION_1:Formulation}, we present the mathematical formulation of the proposed game-theoretic framework. In Section \ref{SECTION_2:Analysis&Properties}, we derive the theoretical properties of the proposed framework, and finally, draw conclusions in Section \ref{SECTION_4:Conclusions}.

\section{Modeling Framework}\label{SECTION_1:Formulation}

We consider a routing game with a finite non-empty set of players $\mathcal{I}$, $|\mathcal{I}| = n \in \mathbb{N}$. Each player $i$ may represents a class of travelers that could use connected and automated vehicles (CAVs) and who control a significant amount of traffic, say $x_i \in \mathbb{R}_{\geq 0}$. Thus, we interpret $x_i$ as the representation of the flow of traffic that player $i$ contributes to a transportation network. We define traffic flow in this setting as the number of vehicles passing through each point in the network over time. This decision variable is non-negative as players (or the travelers) make trips using their vehicle over time in the transportation network. This is in contrast to non-atomic routing games, where players only control an infinitesimal amount of traffic. We also assume that traffic is \emph{splittable}. Travelers seek to travel in a transportation network represented by a directed multigraph $\mathcal{G} = (\mathcal{V}, \mathcal{E})$, where each node in $\mathcal{V}$ may represent different city areas or neighborhoods (e.g., Braess' paradox network). Each edge $e \in \mathcal{E}$ may represent a road. For our purposes, we think of $\mathcal{G} = (\mathcal{V}, \mathcal{E})$ as a representation of a smart city network with a road infrastructure. Any player $i \in \mathcal{I}$ seeks to travel from an origin $o \in \mathcal{V}$ to a destination $d \in \mathcal{V}$. So, all players are associated with the same unique origin-destination pair $(o, d) \in \mathcal{V} \times \mathcal{V}$. Next, each player may use a sequence of edges that connects the OD pair $(o, d)$. We define $\mathcal{R} \subset 2 ^ {\mathcal{E}}$ as the set of routes available to any player $i \in \mathcal{I}$, where their route $r_i$ consists of a sequence of edges connecting the origin-destination pair $(o, d)$. We are interested in how such players may compete over the routes in the network for routing their traffic flows (this is a multiple-route traffic flow decision-making problem).

Since each player $i \in \mathcal{I}$ seeks to route their traffic represented by flow $x_i$ in the network $\mathcal{G}$, we define, for each $i \in \mathcal{I}$, the set of actions as
    \begin{equation}
        \mathcal{X}_i = \left\{ x_i \in \mathbb{R}_{\geq 0} ^ {|\mathcal{R}|}: \sum_{r_i \in \mathcal{R}} x_i ^ {r_i} = \bar{x}_i \right\},
    \end{equation}
where $x_i = (x_i ^ {r_i^{1}}, x_i ^ {r_i^{2}}, \dots, x_i ^ {r_i^{|\mathcal{R}|}})$, $\bar{x}_i \in \mathbb{R}_{\geq 0}$ is the total flow of player $i$, and $r_i ^ k$ denotes the $k$-th route in the network.

\begin{remark}
    Note here that each player $i$ controls their traffic flow $x_i$, which we represent as a vector since player $i$ may choose to use different routes in the transportation network, thus sending traffic $x_i ^ {r_i ^ k}$ for some $k$. The total traffic flow controlled by player $i$ is finite, though. And so, we represent this by introducing $\bar{x}_i$.
\end{remark}

We write $\mathcal{X} = \mathcal{X}_1 \times \mathcal{X}_2 \times \dots \times \mathcal{X}_n$ for the Cartesian product of all the players' action sets. We also write $x_{- i} = (x_1, x_2, \dots, x_{i - 1}, x_{i + 1}, \dots, x_n)$ for the action profile that excludes player $i \in \mathcal{I}$. Next, for the aggregate action profile, we write $x = (x_i, x_{- i})$, $x \in \mathcal{X}$.

\begin{definition}
    The flow on edge $e \in \mathcal{E}$ is the sum of relevant components of all players' traffic flows that have chosen a route that includes edge $e$, i.e., $f_e(x) = \sum_{i \in \mathcal{I}} \sum_{r_i \ni e} x_i ^ {r_i}$.
\end{definition}

In our routing game where each player $i \in \mathcal{I}$ chooses their traffic flow vector $x_i$ over a common set of routes $\mathcal{R}$, if player $i$ chooses to send traffic $x_i ^ {r_i}$ along route $r_i$, then this traffic will be distributed along all the edges in this route $r_i$. This is because a traveler's traffic on some route $r_i$ is a single quantity among all the route's edges.
 
Next, we introduce a \emph{travel time latency function} to capture the cost that players may experience. Intuitively, we capture the players' preferences for different outcomes using a ``cost function," in which players are expected to act as cost minimizers. For each $e \in \mathcal{E}$, we consider non-negative cost functions $c_e : \mathbb{R}_{\geq 0} \to \mathbb{R}_{\geq 0}$. We assume that the cost functions at each edge $e$ are convex, continuous, and differentiable with respect to $f_e$. One standard way to define in an exact form $c_e$ is by the US Bureau of Public Roads (BPR) function, as it is a commonly used model for the relationship between flow and travel time. Mathematically, we have, for any edge $e \in \mathcal{E}$,
    \begin{equation}
        c_e(f_e) = c_e ^ 0 \left( 1 + \frac{3}{20} \left( \frac{f_e}{f_e ^ {CRT}} \right) ^ 4 \right),
    \end{equation}
where $c_e ^ 0$ is the free-flow travel time and $f_e ^ {CRT}$ is the \emph{critical} capacity of traffic flow on road $e$. Note that the BPR function is non-linear, continuous, differentiable, strictly increasing, and strictly convex for $f_e \geq 0$.

\begin{definition}
    If the \emph{maximum flow} on edge $e \in \mathcal{E}$ is $f_e ^ {\max} \in \mathbb{R}_{> 0}$, then for the \emph{critical flow}, $f_e ^ {CRT}$, on edge $e \in \mathcal{E}$ we have $f_e ^ {CRT} < f_e ^ {\max}$.
\end{definition}

Next, for some route $r_i$ of any player $i$, its cost is the sum of the costs on the edges that constitute route $r_i$, i.e., $c_{r_i}(x) = \sum_{e \in r_i} c_e(f_e(x))$. The total cost for player $i$ is
    \begin{equation}
        c_i(x) = \sum_{r_i \in \mathcal{R}} c_{r_i}(x) = \sum_{r_i \in \mathcal{R}} \left[ \sum_{e \in r_i} c_e(f_e) \right],
    \end{equation}
where $f_e(x) = \sum_{i \in \mathcal{I}} \sum_{r_i \ni e} x_i ^ {r_i}$.

The game is fully characterized by the tuple $\mathcal{M} = \langle \mathcal{I}, (\mathcal{X}_i)_{i \in \mathcal{I}}, (c_i)_{i \in \mathcal{I}} \rangle$. This non-cooperative routing game is a simultaneous-move game where players make decisions simultaneously and commute in $(o, d)$ of network $\mathcal{G}$. Players behave selfishly and aim to minimize their costs (e.g., travel time latencies). Naturally, players compete with each other over the available yet limited routes and how to utilize them in the transportation network. Indirectly, players make route choices that satisfy their travel needs (modeled through traffic flow). Next, we clarify ``who knows what?" in $\mathcal{M}$. All players have complete knowledge of the game and the network. Each player knows their own information (action and cost) as well as the information of other players. At equilibrium, we want to ensure that no player has an incentive to unilaterally deviate from their chosen decisions and change how they distribute their traffic flows over the available routes in the network. So, for our purposes, we observe that an NE in terms of the players' traffic flows in pure strategies is the most appropriate a solution concept for our game.

\begin{definition}
    A feasible flow profile $x ^ * = (x_i ^ {r_i})_{i \in \mathcal{I}} ^ {r_i \in \mathcal{R}} \in \mathcal{X}$ constitutes a NE if for each player $i \in \mathcal{I}$, $c_i(x_i ^ *, x_{- i} ^ *) \leq c_i(x_i, x_{- i} ^ *)$, for all $x_i \in \mathcal{X}_i$.
\end{definition}

In other words, a flow profile $x ^ *$ is an NE if no player can reduce their total cost by unilaterally changing how they distribute their total traffic flow over the available routes in the network. In an NE, each player's specific $x_i$ has the lowest possible cost among all possible distributions over the routes, given the choices made by other players.

\subsection{Prospect Theory Analysis}\label{Subsection:NashExistencePT}

In this subsection, we briefly introduce prospect theory and its main concepts \cite{Wakker2010}. Prospect theory attempts to answer one of the main questions of how a decision-maker may evaluate different possible actions/outcomes under uncertain and risky circumstances. Thus, prospect theory is a descriptive behavioral model and focuses on three main behavioral factors: (i) \emph{Reference dependence}: decision-makers make decisions based on their utility, which is measured from the ``gains" or ``losses." However, the utility is a gain or loss relative to a reference point that may be unique to each decision-maker. (ii) \emph{Diminishing sensitivity}: changes in value have a greater impact near the reference point than away from the reference point. (iii) \emph{Loss aversion}: decision-makers are more conservative in gains and riskier in losses. One way to mathematize the above behavioral factors (1) - (3) is to consider an action by a decision-maker as a ``gamble" with objective utility value $z \in \mathbb{R}$. We say that this decision-maker \emph{perceives} $z$ subjectively using a \emph{value function} \cite{Kahneman1979,Al-Nowaihi2008}
    \begin{equation}\label{eqn:value_function}
        v(z) =
            \begin{cases}
                (z - z_0) ^ {\beta_1}, & \text{if } z \geq z_0, \\
                - \lambda (z_0 - z) ^ {\beta_2}, & \text{if } z < z_0,
            \end{cases}
    \end{equation}
where $z_0$ represents a reference point, $\beta_1, \beta_2 \in (0, 1)$ are parameters that represent the diminishing sensitivity. Both $\beta_1, \beta_2$ shape \eqref{eqn:value_function} in a way that the changes in value have a greater impact near the reference point than away from the reference point. We observe that \eqref{eqn:value_function} is concave in the domain of gains and convex in the domain of losses. Moreover, $\lambda \geq 1$ reflects the level of loss aversion of decision-makers. To the best of our knowledge, a widely agreed theory does not exist that determines and defines the reference dependence \cite{Kahneman1979}. In engineering \cite{Hota2016,Etesami2020}, it is assumed that $z_0 = 0$ captures a decision-maker's expected status-quo level of the resources.

Prospect theory models the subjective behavior of decision-makers under uncertainty and risk. Each objective utility $z \in \mathbb{R}$ is associated with a probabilistic occurrence, say $p \in [0, 1]$. Decision-makers are subjective and perceive $p$ differently depending on its value. To capture this behavior, we introduce a strictly increasing function $w : [0, 1] \to [0, 1]$ with $w(0) = 0$ and $w(1) = 1$ called the \emph{probability weighting function}. This function allows us to model how decision-makers may overestimate small probabilities of objective utilities, i.e., $w(p) > p$ if $p$ is close to $0$, or underestimate high probabilities, i.e., $w(p) < p$ if $p$ is close to $1$. We use Prelec's probability weighting function first introduced in \cite{Prelec1998}, $w(p) = \exp \left(- (- \log (p)) ^ {\beta_3} \right), \quad p \in [0, 1]$, where $\beta_3 \in (0, 1)$ represents a \emph{rational index}, i.e., the distortion of a decision-maker's probability perceptions. Mathematically, $\beta_3$ controls the curvature of the weighting function.

\begin{definition}
    Suppose that there are $K \in \mathbb{N}$ possible outcomes available to a decision-maker and $z_k \in \mathbb{R}$ is the $k$-th gain/loss of objective utility. Then a prospect $\ell_k$ is a tuple of the utilities and their respective probabilities, i.e., $\ell_k = (z_0, z_1, z_2, \dots, z_K; p_0, p_1, p_2, \dots, p_K)$, where $k = 0, 1, 2, \dots, K$. We denote the $k$-th prospect more compactly as $\ell_k = (z_k, p_k)$. We have that $\sum_{k = 0} ^ K p_k = 1$ and $\ell_k$ is well-ordered, i.e., $z_0 \leq z_1 \leq \cdots \leq z_K$. Under prospect theory, the decision-maker evaluates their ``subjective utility" as $u(\ell) = \sum_{0 \leq k \leq K} v(z_k) w(p_k)$, where $\ell = (\ell_k)_{k = 1} ^ K$ is the profile of prospects of $K$ outcomes.
\end{definition}

In the remainder of this subsection, we apply the prospect theory to our modeling framework, clearly define the mobility outcomes (objective and subjective utilities), and then show that the prospect-theoretic game $\mathcal{M}$ admits a NE.

Players may be uncertain about the value of the traffic disturbances as it is affected by unexpected factors, and so we use Prelec's probability weighting function $w : [0, 1] \to [0, 1]$ to capture how different traveler populations ``perceive" probabilities. In addition, we are interested in capturing how players may perceive their gains or losses regarding their travel time costs with respect to the costs at critical density. Hence, we define the \emph{mobility prospect} as whether $f_e$ will reach its critical or jammed point. Formally, $\pi_e$ is the probability that $f_e \in (0, f_e ^ {CRT})$, and $1 - \pi_e$ is the probability for $f_e \in (f_e ^ {CRT}, f_e ^ {\max}]$. We then use the prospect-theoretic S-shaped value function $v(c_e(f_e)): \mathbb{R}_{\geq 0} \to \mathbb{R}$ to capture how players may perceive such costs. Hence, we have
    \begin{equation}\label{eqn:valuation_prospect}
        v(c_e(f_e)) =
            \begin{cases}
                \lambda (c_e ^ 0 - c_e(f_e)) ^ {\beta}, & \text{if } c_e(f_e) \leq c_e ^ 0, \\
                - (c_e(f_e) - c_e ^ 0) ^ {\beta}, & \text{if } c_e(f_e) > c_e ^ 0,
            \end{cases}
    \end{equation}
where the reference dependence is represented by $c_e ^ 0 = c_e(f_e ^ {CRT})$, $\beta_1 = \beta_2 = \beta \in (0, 1)$, and for each $e \in \mathcal{E}$, we have $\pi_e \in [0, 1]$. We justify $\beta_1 = \beta_2$ in above as it has been verified to produce extremely good results, and the outcomes are consistent with the original data \cite{Tversky1992}. We define
    \begin{equation}
        \tilde{c}_e(f_e) = 
            \begin{cases}
                c_e ^ 0 - c_e(f_e), & \text{if } c_e(f_e) < c_e ^ 0, \\
                c_e(f_e) - c_e ^ 0, & \text{if } c_e(f_e) > c_e ^ 0.
            \end{cases}
    \end{equation}

\begin{remark}
    It is important to note that our prospect-theoretic value function is ``reversed," capturing the way a traveler will perceive the gains in travel time through a cost function. Using as a reference point the critical traffic flow on edge $e$, we can pinpoint the exact point that any more delays become socially unacceptable, i.e., a higher flow causes a higher travel time that a traveler will not tolerate.
\end{remark}

The new cost function is
    \begin{multline}
        c_e ^ {\text{PT}} (f_e) = w(\pi_e) \cdot \lambda \cdot [ \tilde{c}_e(f_e \, | \, c_e(f_e) < c_e ^ 0) ] ^ {\beta} \\
        - w(1 - \pi_e) \cdot [ \tilde{c}_e(f_e \, | \, c_e(f_e) > c_e ^ 0) ] ^ {\beta}.
    \end{multline}
Observe that this particular formulation allows only two main outcomes for any player. One outcome may represent an easy commute (no traffic), and the other may represent traffic. For our purposes, we naturally expect two probabilities for these two outcomes. Future work will ensure to allow a larger distribution of probabilities for many different outcomes for any player (in such cases, cumulative prospect theory would be a more appropriate model \cite{Tversky1992}).

The total cost on some route $r_i$ for player $i$ under prospect theory is $c_{r_i} ^ {\text{PT}} (x) = \sum_{e \in r_i} c_e ^ {\text{PT}} (f_e)$. Now, the total cost of some player $i$ is given by
    \begin{equation}\label{eqn:piecewis}
        c_i ^ {\text{PT}} (x) = \sum_{r_i \in \mathcal{R}} \sum_{e \in r_i} c_e ^ {\text{PT}} (f_e).
    \end{equation}
Note, however, that in this case, the prospect-theoretic cost is capturing the gains and losses of travel. Thus, the aim is to maximize this function to maximize the gains (by minimizing the actual cost of travel latencies).

What we observe in \eqref{eqn:piecewis} is that it is rather cumbersome to analyze it analytically as issues in its smoothness arise quickly. The problem in analyzing such a function is that the exponent takes values in $(0, 1)$. To address this theoretical obstacle, we propose a new function that approximates the prospect-theoretic function and, most importantly, can be shown to have useful properties. Hence, we define the following function
    \begin{equation}\label{eqn:approx}
        \sigma(f_e) = \frac{\delta_1}{1 + \exp{ \left( \frac{\delta_2 - f_e}{\delta_3} \right)}} + \delta_4,
    \end{equation}
where $\delta_1, \delta_2, \delta_3, \delta_4 \in \mathbb{R}$, and $f_e \in [0, f_e ^ {\max} ]$. Hence, we can approximately evaluate \eqref{eqn:piecewis} with the following:
    \begin{equation}
        c_i ^ {\text{PT}} (x) = \sum_{r_i \in \mathcal{R}} \sum_{e \in r_i} \sigma(f_e).
    \end{equation}

\section{Analysis and Properties of the Game}\label{SECTION_2:Analysis&Properties}

In this section, we provide a formal analysis of the properties of our proposed modeling framework, characterize the coefficients of $\sigma$ function, and show that our game admits an NE in pure strategies.

\begin{lemma}\label{lemma:strategy_set}
    The strategy space of the game $\mathcal{M}$ is non-empty, compact, and convex.
\end{lemma}

\begin{proof}
   The proof has been omitted here due to space-constraints.
\end{proof}

\begin{lemma}\label{lemma:concavity}
    The approximation function given by \eqref{eqn:approx} in the interval $[0, \kappa]$, $\kappa < f_e ^ {\max}$, is strictly concave with respect to $f_e$ when $\delta_3 > 0$, $\delta_4 \in \mathbb{R}$, and (i) $\delta_1 > 0$, $\delta_2 > f_e$, or alternatively (ii) $\delta_1 < 0$, $\delta_2 < f_e$.
\end{lemma}

\begin{proof}
    Given that $f_e \geq 0$, we analyze the second-order derivative of the function $\sigma(f_e) = \frac{\delta_1}{1 + \exp{ \left( \frac{\delta_2 - f_e}{\delta_3} \right) }} + \delta_4$ to determine the conditions for strict concavity. First, let us find the first and second-order derivatives of $\sigma$ with respect to $f_e$, i.e.,
        \begin{align}
            \sigma ' (f_e) & = \frac{- \delta_1 \exp{(\frac{\delta_2 - f_e}{\delta_3})}}{\delta_3 (1 + \exp{(\frac{\delta_2 - f_e}{\delta_3})}) ^ 2}, \\
            \sigma '' (f_e) & = \frac{2 \delta_1 \exp{(\frac{\delta_2 - f_e}{\delta_3})}(1 - \exp{(\frac{\delta_2 - f_e}{\delta_3})})}{\delta_3 ^ 2(1 + \exp{(\frac{\delta_2 - f_e}{\delta_3})}) ^ 3}.
        \end{align}
    Now, we examine the conditions for $\sigma '' (f_e) < 0$. First, $\delta_1$ controls the sign of the second-order derivative as follows: if $\delta_1 < 0$ and $\delta_3 > 0$, $\sigma '' (f_e)$ will be negative when $1 - \exp{(\frac{\delta_2 - f_e}{\delta_3})} < 0$, which simplifies to $\delta_2 > f_e$. If $\delta_3 < 0$ in either of the cases, then the signs are reversed. We do require though that $\delta_3 ^ 2$ is well-defined, so $\delta_3 \neq 0$. On greater detail, $1 - \exp{(\frac{\delta_2 - f_e}{\delta_3})}$ determines the conditions for $\sigma '' (f_e)$ to be negative. If $\delta_1 < 0$, we need $1 - \exp{(\frac{\delta_2 - f_e}{\delta_3})} > 0$, which implies that $f_e > \delta_2 - \delta_3 \log{(1)}$ (since $f_e \geq 0$). If $\delta_1 < 0$, we need $1 - \exp{(\frac{\delta_2 - f_e}{\delta_3})} < 0$, which implies that $f_e < \delta_2 - \delta_3 \log{(1)}$.
    
    Combining these insights, we can conclude that the function $\sigma(f_e) = \frac{\delta_1}{1 + \exp{(\frac{\delta_2 - f_e}{\delta_3})}} + \delta_4$ becomes strictly concave in the entire interval. So, it is strictly concave for $f_e \geq 0$ if: (i) $\delta_1 > 0, \delta_3 > 0$ and $f_e < \delta_2$; (ii) $\delta_1 < 0$, $\delta_3 > 0$, and $f_e > \delta_2$. If $\delta_3 < 0$, then the relation between $f_e$ and $\delta_2$ is naturally reversed. Note that the parameter $\delta_4$ does not affect the convexity of the function, as it only shifts the function vertically. Therefore, we have derived the necessary conditions that ensure $\sigma '' (f_e)$ is negative for all $f_e$, making $\sigma(f_e)$ strictly concave.
\end{proof}

It follows that it is strictly decreasing, continuous, and (continuously) differentiable with respect to the traffic flow $f_e \in [0, f_e ^ {\max}]$ for any edge $e \in \mathcal{E}$.

Now, we discuss the error characterization of our approximation function. Let us define the error function $\Phi$ as the squared difference between $c_e ^ \text{PT} (f_e)$ and $\sigma(f_e)$, integrated over the interval $[0, \kappa]$:
    \begin{equation}
        \Phi(\delta_1, \delta_2, \delta_3, \delta_4) = \int_{0} ^ {\kappa} \left( c_e ^ \text{PT} (f_e) - \sigma(f_e) \right) ^ 2 \mathrm{d}x.
    \end{equation}
The goal is to minimize $\Phi$ with respect to the parameters $\delta_1, \delta_2, \delta_3$, and $\delta_4$. First, we find the critical points of $\Phi$ by setting its gradient to zero and solving the resulting system of equations: $\nabla \Phi(\delta_1, \delta_2, \delta_3, \delta_4) = 0$. This results in a system of equations involving the partial derivatives of $\Phi$ with respect to each of the parameters, i.e., $\frac{\partial \Phi}{\partial \delta_1} = 0, \quad \frac{\partial \Phi}{\partial \delta_2} = 0, \frac{\partial \Phi}{\partial \delta_3} = 0$, and $\frac{\partial \Phi}{\partial \delta_4} = 0$. To compute these partial derivatives, we need to differentiate the integrand with respect to each parameter and then integrate it again, for example, $\frac{\partial \Phi}{\partial \delta_1} = \int_0 ^ {\kappa} \frac{\partial}{\partial \delta_1} \left( c_e ^ \text{PT} (f_e) - \sigma(f_e) \right) ^ 2 \mathrm{d}x$. This process needs to be repeated for all parameters. However, due to the complexity of the function $c_e ^ \text{PT} (f_e)$ (being a non-linear piecewise function), it is not possible to obtain an explicit analytical expression for these partial derivatives. For our purposes, we rely on numerical optimization techniques to find the exact best-fit parameters that minimize the error function, as these methods can easily handle complex and non-linear optimization.

\begin{theorem}
    The error $\phi(\cdot) = \left( c_e ^ {\text{PT}} (f_e) - \sigma(f_e) \right)$ is upper bounded by $\gamma + \varepsilon$, where $\gamma$ is some real number and $\varepsilon > 0$.
\end{theorem}

\begin{proof}
    For the purposes of this proof we assume that $\beta = 0.5$, $\lambda = 2$ and $c_e ^ 0 = 13$ and $f_e ^ {\text{CRT}} = 1$ and $c_e ^ 0 = 14.95$. We substitute now the known equations to get $\phi(\delta_1, \delta_2, \delta_3, \delta_4) = - w(1 - \pi_e) \cdot [ \tilde{c}_e(f_e \, | \, c_e(f_e) > c_e ^ 0) ] ^ {\beta} - \frac{\delta_1}{1 + \exp{ \left(\frac{\delta_2 - f_e}{\delta_3} \right)}} - \delta_4$. Using a straightforward computation of the second-order derivative, we can get the inflection point of $\phi$, which will lie in $(1, 1 + \varepsilon)$. This means that it is sufficient for us to compute $\phi$ at $f_e = 1$ and focus on $\phi$ for $f_e > 1$. Since $\sigma$ is smooth and strictly concave in that interval, it approximates the worst $c_e ^ {\text{PT}}$ around the inflection point. So, we have the following $\phi(\delta_1, \delta_2, \delta_3, \delta_4) = - w(1 - \pi_e) (c_e(f_e) - c_e ^ 0) ^ {\beta} - \frac{\delta_1}{1 + \exp{ \left(\frac{\delta_2 - f_e}{\delta_3} \right)}} - \delta_4$. This expression simplifies to
        \begin{multline}
            \phi = - w(1 - \pi_e) \left[ 13 \Bigg( 1 + \frac{3}{20} (f_e) ^ 4 \Bigg) - \frac{299}{20} \right] ^ {\beta} \\
            - \frac{\delta_1}{1 + \exp{ \left(\frac{\delta_2 - f_e}{\delta_3} \right)}} - \delta_4,
        \end{multline}
    where we have $\delta_1 < 0$ and $\delta_2, \delta_3, \delta_4 > 0$, and $\delta_2 > f_e$. Since $w(1 - \pi_e)$ is only a positive parameter constant, it is negligible, so we drop it from our analysis. The first component simplifies to $\left[ \frac{39}{20} \left( 1 - (f_e) ^ 4 \right) \right] ^ {\beta}$, which is negative when we evaluate near the inflection point. Next, it follows that the second component is positive for small values of $\delta_2$ and $\delta_3$. We use the Taylor series expansion evaluated at $f_e = 1 + \varepsilon$, where $\varepsilon$ is a small positive number to get 
        \begin{multline}
            - \left[ \frac{39}{20} \left( (f_e) ^ 4 - 1 \right) \right] ^ {\beta} = \\
            - \sqrt{\frac{39}{5}} \sqrt{\varepsilon} - \frac{3}{4} \frac{\sqrt{39}}{5} \varepsilon ^ {\frac{3}{2}} - \frac{7}{32} \frac{\sqrt{39}}{5} \varepsilon ^ {\frac{5}{2}} + O(\varepsilon ^ {\frac{7}{2}}),
        \end{multline}
    which is clearly negative. For the second component, we use the Taylor series expansion at $f_e = 1 + \varepsilon$ to get
        \begin{multline}
            \frac{\delta_1}{1 + \exp{ \left(\frac{\delta_2 - (1 + \varepsilon)}{\delta_3} \right)}} = \frac{\delta_1}{1 + \exp{\left(\frac{\delta_2 - 1}{\delta_3}\right)}} \\
            + \frac{\delta_1 \exp{\left(\frac{\delta_2 - 1}{\delta_3}\right)}}{\delta_3 \left(1 + \exp{\left(\frac{\delta_2 - 1}{\delta_3}\right)}\right)^2} \cdot \varepsilon \\
            \frac{1}{2} \frac{\delta_1 \exp{\left(\frac{1}{\delta_3} + \frac{\delta_2}{\delta_3}\right)} \left(\exp{\left(\frac{1}{\delta_3}\right)} - \exp{\left(\frac{\delta_2}{\delta_3}\right)}\right)}{\delta_3^2 \left(\exp{\left(\frac{1}{\delta_3}\right)} + \exp{\left(\frac{\delta_2}{\delta_3}\right)}\right)^3} \cdot \varepsilon^2 + O(\varepsilon^3).
        \end{multline}
    We combine the expressions for the first and second components. Next, we have
        \begin{multline}\label{eqn:phi_expression}
            \phi = - \sqrt{\frac{39}{5}} \sqrt{\varepsilon} - \frac{3}{4} \frac{\sqrt{39}}{5} \varepsilon ^ {\frac{3}{2}} - \frac{7}{32} \frac{\sqrt{39}}{5} \varepsilon ^ {\frac{5}{2}} + O(\varepsilon ^ {\frac{7}{2}}) \\
            + \frac{\delta_1}{1 + \exp{\left(\frac{\delta_2 - 1}{\delta_3}\right)}} + \frac{\delta_1 \exp{\left(\frac{\delta_2 - 1}{\delta_3}\right)}}{\delta_3 \left(1 + \exp{\left(\frac{\delta_2 - 1}{\delta_3}\right)}\right)^2} \cdot \varepsilon \\
            + \frac{1}{2} \frac{\delta_1 \exp{\left(\frac{1}{\delta_3} + \frac{\delta_2}{\delta_3}\right)} \left(\exp{\left(\frac{1}{\delta_3}\right)} - \exp{\left(\frac{\delta_2}{\delta_3}\right)}\right)}{\delta_3^2 \left(\exp{\left(\frac{1}{\delta_3}\right)} + \exp{\left(\frac{\delta_2}{\delta_3}\right)}\right)^3} \cdot \varepsilon^2 + O(\varepsilon^3).
        \end{multline}
    We want to find an upper bound for the error, which means we need to show that \eqref{eqn:phi_expression} is less than or equal to $\gamma + \varepsilon$ for some $\gamma \in \mathbb{R}$. Note that for any $a, b \in \mathbb{R}$ with $a < 0$ and $b > 0$, it is always true that $a + b \leq \max\{a, b\}$. So,
        \begin{multline}
            \phi \leq \max \bigg\{ - \sqrt{\frac{39}{5}} \sqrt{\varepsilon} - \frac{3}{4} \frac{\sqrt{39}}{5} \varepsilon ^ {\frac{3}{2}} - \frac{7}{32} \frac{\sqrt{39}}{5} \varepsilon ^ {\frac{5}{2}} + O(\varepsilon ^ {\frac{7}{2}}), \\
            \frac{\delta_1}{1 + \exp{\left(\frac{\delta_2 - 1}{\delta_3}\right)}} + \frac{\delta_1 \exp{\left(\frac{\delta_2 - 1}{\delta_3}\right)}}{\delta_3 \left(1 + \exp{\left(\frac{\delta_2 - 1}{\delta_3}\right)}\right)^2} \cdot \varepsilon \\
            + \frac{1}{2} \frac{\delta_1 \exp{\left(\frac{1}{\delta_3} + \frac{\delta_2}{\delta_3}\right)} \left(\exp{\left(\frac{1}{\delta_3}\right)} - \exp{\left(\frac{\delta_2}{\delta_3}\right)}\right)}{\delta_3^2 \left(\exp{\left(\frac{1}{\delta_3}\right)} + \exp{\left(\frac{\delta_2}{\delta_3}\right)}\right)^3} \cdot \varepsilon^2 + O(\varepsilon^3) \bigg\}.\nonumber
        \end{multline}
    As $\varepsilon$ is positively small, we take the limit as $\varepsilon \to 0$. We note that the term $- \sqrt{\frac{39}{5}} \sqrt{\varepsilon}$ dominates as $\varepsilon \to 0$, and so the first component approaches $-\infty$ as $\varepsilon \to 0$. For the second component, the term $\frac{\delta_1}{1 + \exp{\left(\frac{\delta_2 - 1}{\delta_3}\right)}}$ dominates as $\varepsilon \to 0$, and since $\delta_1 < 0$ and $\delta_2, \delta_3 > 0$, the second component is positive. Hence, we can write
        \begin{equation}
            \lim_{\varepsilon \to 0} \phi \leq \lim_{\varepsilon \to 0} \max \left\{ - \sqrt{\frac{39}{5}} \sqrt{\varepsilon}, \frac{\delta_1}{1 + \exp{\left(\frac{\delta_2 - 1}{\delta_3}\right)}} \right\}.
        \end{equation}
    As $\varepsilon \to 0$, we have $- \sqrt{\frac{39}{5}} \sqrt{\varepsilon} \to -\infty$, hence
        \begin{align}
            \lim_{\varepsilon \to 0} \phi & \leq \lim_{\varepsilon \to 0} \max \left\{- \sqrt{\frac{39}{5}} \sqrt{\varepsilon}, \frac{\delta_1}{1 + \exp{\left(\frac{\delta_2 - 1}{\delta_3}\right)}} \right\} \notag \\
            & = \frac{\delta_1}{1 + \exp{\left(\frac{\delta_2 - 1}{\delta_3}\right)}}.
        \end{align}
    Now, let $\gamma = \frac{\delta_1}{1 + \exp{\left(\frac{\delta_2 - 1}{\delta_3}\right)}}$. Since the second component is positive, we have $\gamma > 0$, thus $\phi \le \gamma + \varepsilon$. Therefore, the error $\phi$ is upper bounded by $\gamma + \varepsilon$, where $\gamma \in \mathbb{R}$ and $\varepsilon > 0$.
\end{proof}

\begin{theorem}
    The game $\mathcal{M}$ admits at least one NE.
\end{theorem}

\begin{proof}
    We formally prove the existence of an NE in the prospect-theoretic routing game using Brouwer's fixed point theorem. Recall that for any player $i$, $c_i ^ {\text{PT}} (x) = \sum_{r_i \in \mathcal{R}} \sum_{e \in r_i} \sigma(f_e(x))$, where $\sigma$ is our smooth and monotonic approximation function. We define the best-response correspondence for each player $i$ as: $b_i(x_{-i}) = \arg\max_{x_i} c_i^{\text{PT}}(x)$. Smoothness in the approximation function $\sigma$ implies that it is continuous and has continuous derivatives. This implies that we can estimate the utility function $c_i ^ {\text{PT}}(x)$ continuously with respect to the traffic vector $x$. To show that the best-response correspondence $b_i(x_{-i})$ is continuous, we need the $\arg\max$ operator to be continuous. Since the set of maximizers is compact, which actually follows from the compactness of the strategy space by Lemma \ref{lemma:strategy_set}. By Lemma \ref{lemma:concavity}, we have that $\sigma$ is concave on a specific interval $[0, \kappa]$. This implies that we can estimate the utility function $c_i ^ {\text{PT}}(x)$ within the interval $[0, \kappa]$ pointwise in a strictly decreasing and strictly concave curve with respect to $x_i$ for any player $i \in \mathcal{I}$. However, a strictly concave function has at most one unique maximum, which ensures the single-validness of the best-response correspondence $b_i(x_{-i})$. We now define the combined best-response correspondence $B(x) = (b_1(x_{-1}), b_2(x_{-2}), \dots, b_n(x_{-n}))$. Since each $b_i(x_{-i})$ is continuous, $B(x)$ is also continuous, and thus it maps the strategy space to itself. Hence, now we can apply Brouwer's fixed point theorem, which guarantees that there exists a fixed point $x^* = B(x^*)$; the result then follows.
\end{proof}

\addtolength{\textheight}{-3cm}   % This command serves to balance the column lengths

\section{CONCLUSIONS AND FUTURE WORK}\label{SECTION_4:Conclusions}

%\subsection{Conclusions}

In this paper, we presented our mobility game incorporating an atomic splittable routing game with prospect theory to study travel behavior in mobility systems. We modeled the overestimation/underestimation of probabilities using Prelec's probability weighting function, and we considered the traffic uncertainties and travelers' perception of gains/losses in travel time using a prospect-theoretic S-shaped value function. We proposed an approximation function to address the non-linear and piecewise nature of the prospect-theoretic cost functions and showed that at least one NE exists. Finally, we also derived an upper bound for the error.

%\subsection{Future Works}

In future research, we can explore how to analyze a convex-concave piecewise non-linear optimization problem using optimization techniques, such as sequential convex programming or cutting plane methods. Developing such an optimization framework can enhance our ability to predict travel decisions in mobility systems under prospect theory. Another direction is to incorporate prospect theory and a taxation mechanism and study using artificial intelligence how we can incentivize prospect-theoretic travelers and the trade-offs of efficiency in the mobility systems \cite{Chremos2020MechanismDesign,chremos2022Equity,bharadiya2023artificial}.

\bibliographystyle{IEEEtran}
\bibliography{References/IDS_Publications_03122024,References/MyReferences}

\end{document}